\documentclass[11pt,english,numberwithinsect]{article}

\usepackage[left=1in,right=1in,top=1in,bottom=1in]{geometry}

\usepackage{amsthm}
\usepackage{amsmath}
\usepackage{amssymb}

\usepackage[noend]{algpseudocode}
\usepackage{algorithm}

\usepackage{graphicx}
\usepackage{calc}

\usepackage{tikz}

\usepackage{subcaption}

\usepackage{color}

\usepackage{wrapfig}

\usepackage{xspace}

\usepackage[bookmarks]{hyperref}

\usepackage{varioref} 

\usepackage{footnote} 

\usepackage{enumitem} 

\mathchardef\mhyphen="2D

\newcommand{\eps}{\ensuremath{\varepsilon}}

\renewcommand{\algref}[1]{Algorithm~\ref{alg:#1}}

\newcommand{\proporef}[1]{Proposition~\ref{propo:#1}}
\newcommand{\thmref}[1]{Theorem~\ref{thm:#1}}

\newcommand{\lemref}[1]{Lemma~\ref{lem:#1}}

\newcommand{\secref}[1]{Section~\ref{sec:#1}}

\newtheorem{thm}{Theorem}[section]
\newtheorem{lem}[thm]{Lemma}

\newtheorem{proposition}[thm]{Proposition}

\newtheorem{claim}[thm]{Claim}

\global\long\def\Oh{{\cal O}}
\global\long\def\tOh{\tilde{\Oh}}
\global\long\def\tOmega{\tilde{\Omega}}

\newcommand{\SSB}[2]{\mathcal{S}(#1;#2)}
\newcommand{\SSunb}[2]{\mathcal{S}^{unb}(#1;#2)}
\newcommand{\SUM}{\Sigma}
\newcommand{\pl}[1][]{\oplus_{#1}}
\newcommand{\plt}{\pl[t]}
\newcommand{\N}{\mathbb{N}}
\newcommand{\Z}{\mathbb{Z}}
\newcommand{\Ex}{\mathbb{E}}

\newcommand{\poly}{\textup{poly}}
\newcommand{\polylog}{\textup{polylog}}

\newcommand{\F}{\mathcal{F}}
\newcommand{\C}{\mathcal{C}}
\newcommand{\vc}[1]{\mathbf{#1}}
\newcommand{\va}{\vc{a}}
\newcommand{\vb}{\vc{b}}
\newcommand{\ve}{\vc{e}}

\newcommand{\vf}{\vc{f}}
\newcommand{\outC}{\textup{out}(C)}
\newcommand{\out}[1]{\textup{out}(#1)}

\newcommand{\len}{\textup{len}}
\newcommand{\summ}{\textup{sum}}

\newcommand{\bigboxplus}{
  \mathop{
    \vphantom{\bigoplus} 
    \mathchoice
      {\vcenter{\hbox{\resizebox{\widthof{$\displaystyle\bigoplus$}}{!}{$\boxplus$}}}}
      {\vcenter{\hbox{\resizebox{\widthof{$\bigoplus$}}{!}{$\boxplus$}}}}
      {\vcenter{\hbox{\resizebox{\widthof{$\scriptstyle\oplus$}}{!}{$\boxplus$}}}}
      {\vcenter{\hbox{\resizebox{\widthof{$\scriptscriptstyle\oplus$}}{!}{$\boxplus$}}}}
  }\displaylimits 
}

\newcommand{\kClique}{$k$\nobreakdash-\textsc{Clique}\xspace}
\newcommand{\SetCover}{\textsc{SetCover}\xspace}
\newcommand{\SubsetSum}{\textsc{SubsetSum}\xspace}
\newcommand{\UnbSubsetSum}{\textsc{UnboundedSubsetSum}\xspace}

\newcommand{\ColorCoding}{\mathtt{ColorCoding}}
\newcommand{\ColorCodingLayer}{\mathtt{ColorCodingLayer}}

\newcommand{\FasterSubsetSum}{\mathtt{FasterSubsetSum}}

\newcommand{\IfSODA}[1]{}
\newcommand{\IfNotSODA}[1]{#1}

\IfSODA{\newcommand{\qed}{\hfill \ensuremath{\Box}}}

\begin{document}

\title{A Near-Linear Pseudopolynomial Time Algorithm for Subset Sum}
\author{Karl Bringmann\thanks{Max Planck Institute for Informatics, Saarland Informatics Campus, Germany, \texttt{kbringma@mpi-inf.mpg.de}. 
} 
}
\maketitle

\medskip

\begin{abstract}
Given a set $Z$ of $n$ positive integers and a target value~$t$, the \SubsetSum problem asks whether any subset of~$Z$ sums to $t$. A textbook pseudopolynomial time algorithm by Bellman from 1957 solves \SubsetSum in time $\Oh(n \, t)$. This has been improved to $\Oh(n \max Z)$ by Pisinger [J.~Algorithms'99] and recently to $\tOh(\sqrt{n} \, t)$ by Koiliaris and Xu [SODA'17].

Here we present a simple randomized algorithm running in time $\tOh(n+t)$. This improves upon a classic algorithm and is likely to be near-optimal, since it matches conditional lower bounds from \SetCover and \kClique. 

We then use our new algorithm and additional tricks to improve the best known polynomial space solution from time $\tOh(n^3 t)$ and space $\tOh(n^2)$ to time $\tOh(n \, t)$ and space $\tOh(n \log t)$, assuming the Extended Riemann Hypothesis. Unconditionally, we obtain time $\tOh(n \, t^{1+\eps})$ and space $\tOh(n \, t^\eps)$ for any constant $\eps > 0$.

\end{abstract}

\IfNotSODA{\newpage}


\section{Introduction}



Given a (multi-)set $Z$ of $n$ positive integers and a target~$t$, the \SubsetSum problem asks whether there is a subset $Y$ of $Z$ summing to exactly~$t$.  
This classic NP-hard problem draws its significance from the fact that it lies at the core of many other NP-hard (optimization) problems, e.g. \SubsetSum can easily be reduced to the \textsc{Knapsack} problem, \textsc{Constrained Shortest Path} (also known as \textsc{Restricted Shortest Path}), and many other problems. 

Bellman showed in 1957 that \SubsetSum can be solved in pseudopolynomial time $\Oh(n t)$ by dynamic programming~\cite{bellman1957dynamic}. This algorithm is being taught for decades in undergraduate algorithms courses and thus had a great influence on computer science.
Since this pseudopolynomial time algorithm is a fundamental part of our curriculum, and \SubsetSum is one of the core NP-hard problems, it is an important question whether the running time of $\Oh(nt)$ can be improved.

Pisinger~\cite{pisinger2003dynamic} used basic RAM parallelism (allowing to operate on $\Theta(\log t)$ bits in constant time) to obtain the first improvement: an $\Oh(nt/\log t)$ algorithm. Considering $s := \max Z$, Pisinger~\cite{pisinger1999linear} presented an $\Oh(n s)$ algorithm, which is faster than $\Oh(nt)$ in some situations.
The first algorithm breaking through the $\Oh(nt)$ barrier by a polynomial factor in the worst case was the recent $\tOh(\sqrt{n} \, t)$ algorithm\footnote{For a running time $T$ depending on the input size $n$ and possibly more parameters, such as $t$, we write $\tOh(T)$ as shorthand for $\Oh(T \polylog(T))$. Similarly, $\tOmega(T)$ denotes $\Omega(T / \polylog(T))$.} by Koiliaris and Xu~\cite{KoiliarisX15}. Roughly speaking, they use hashing to solve \SubsetSum quickly if $Z$ is contained in a small interval. Then they reduce the general case to the small interval case by appropriately splitting $Z$ into smaller sets. They also presented an $\tOh(t^{4/3})$ algorithm.

There is reason to believe that \SubsetSum has no $t^{1-\eps} n^{\Oh(1)}$ algorithm for any $\eps > 0$, since this would yield an $2^{(1-\eps')n} (nm)^{\Oh(1)}$ algorithm for \textsc{SetCover} (on $m$ sets over a universe of size $n$) via the reductions in~\cite{cygan2016problems}, and thus break the \textsc{SetCover} hypothesis. A second reason is that any combinatorial $t^{1-\eps} n^{\Oh(1)}$ algorithm for \SubsetSum would yield a combinatorial $\Oh(n^{(1-\eps')k})$ algorithm for \textsc{$k$-Clique}, for some large constant~$k$, via the reduction in~\cite{abboud2014losing}\footnote{Abboud et al.~\cite{abboud2014losing} presented a reduction from \textsc{$k$\nobreakdash-Clique} to \textsc{$\Oh(k^2)$-SUM} on $\Oh(f(k) n^2)$ integers in the range $\{1,\ldots,\Oh(f(k) (n^{1+o(1)})^k)\}$ for some (explicit) function $f$. Note that \textsc{$k'$-SUM} can be reduced to \SubsetSum by introducing $\Oh(\log k')$ leading bits that ensure choosing exactly $k'$ integers. This yields a reduction from \textsc{$k$-Clique} to \SubsetSum on $\Oh(f'(k) n^2)$ integers in $\{1,\ldots,\Oh(f'(k) (n^{1+o(1)})^k)\}$ for some function $f'$. Let $\eps > 0$, $c \in \mathbb{R}$ be constants and let $k$ be sufficiently large ($k \ge 8c/\eps$). Then for all sufficiently large $n \ge n_0$, such that the $n^{o(1)}$-factor is less than $n^{\eps/4}$, a combinatorial $\Oh(t^{1-\eps} n^c)$ algorithm for \SubsetSum would yield a combinatorial $\Oh(n^{k-\eps k/2})$ algorithm for $k$-Clique, refuting the combinatorial $k$-Clique conjecture. 
}. 
This means that polynomial improvements over running time $t$, even in terms of $n$, are unlikely. However, these arguments give no evidence that the additional factor $n$ of the $\Oh(nt)$ dynamic programming algorithm is necessary (or the factor $\sqrt{n}$ of Koiliaris and Xu).
Specifically, they leave the open problem whether there is an $\tOh(n+t)$ algorithm.

Note that an $\tOh(n+t)$ algorithm would also up to polylogarithmic factors dominate the $\Oh(n s)$ algorithm by Pisinger~\cite{pisinger1999linear}, where $s = \max Z$, since any non-trivial instance satisfies $s \ge t/n$.

\paragraph{Near-linear Time}
We present an algorithm in time $\tOh(n+t)$. This improves the classic dynamic programming solution by a factor $\tilde \Omega(n)$ and the best known algorithm by a factor $\tilde \Omega(\sqrt{n})$. Moreover, we match the conditional lower bounds, so any further polynomial improvement would beat the $\Oh^*(2^n)$ \textsc{SetCover} algorithm and possibly the $\Oh(n^k)$ combinatorial \textsc{$k$-Clique} algorithm. 

\begin{thm}[\secref{simple}] \label{thm:main}
  \SubsetSum can be solved in time $\tOh(n+t)$ by a randomized, one-sided error algorithm with error probability $(n+t)^{-\Omega(1)}$. 
\end{thm}

More precisely, if $\SSB{Z}{t}$ is the set of all sums generated by subsets of $Z$ and bounded by $t$, then we can compute in time $\Oh(n+t \log t \log^3(n/\delta) \log n)$ a set $S \subseteq \SSB{Z}{t}$ containing any $s \in \SSB{Z}{t}$ with probability at least $1-\delta$. 


We briefly also consider the \UnbSubsetSum problem, where each input number may be used multiple times instead of just once, and present a simple deterministic $\tOh(n+t)$ algorithm. In \secref{unbounded} we prove this result and discuss why it is much simpler than our main result.

\begin{thm} \label{thm:unbsubsetsum}
  \UnbSubsetSum can be solved deterministically in time $\Oh(n+t \log t)$.
\end{thm}

\paragraph{Polynomial Space}
A surprising, relatively recent result is that \SubsetSum can be solved in pseudopolynomial time and \emph{polynomial} space: Lokshtanov and Nederlof~\cite{lokshtanov2010saving} presented an algorithm in time $\tOh(n^3 t)$ and space $\tOh(n^2)$. In contrast, the $\Oh(nt)$ dynamic programming algorithm as well as all improvements need space $\tilde \Theta(t)$, which can be much larger than $\poly(n)$. 
We use our new \SubsetSum algorithm and some additional tricks to improve upon their bounds. This almost answers an open problem by J.~Nederlof~\cite{husfeldtDagstuhl} for a time $\tOh(nt)$ and space $\tOh(n)$ algorithm.

\begin{thm}[\secref{polynomial}] \label{thm:polynomial}
  \SubsetSum has a randomized, one-sided error algorithm with error probability $(n+t)^{-\Omega(1)}$ in 
  \begin{itemize}
    \item time $\tOh(n \, t)$ and space $\tOh(n \log t)$ assuming the Extended Riemann Hypothesis (ERH), and 
    \item time $\tOh(n \, t^{1+\eps})$ and space $\tOh(n \, t^\eps)$  for any constant $\eps > 0$ unconditionally.
  \end{itemize}
\end{thm}

We leave it as an open problem to find an algorithm with time $\tOh(n+t)$ and space $\tOh(n \log t)$.

\subsection{Techniques} \label{sec:techniques}

Our techniques are very different from previous improvements for \SubsetSum. In particular, all previous algorithms are deterministic.
Our near-linear time algorithm is simple and elegant and makes extensive use of randomization and the fast Fourier transform (FFT). 
In the following we discuss the main ingredients of our algorithms.

\paragraph{Sumset Computation}
For sets of integers $A,B$ we define the \emph{sumset} $A \pl B$ as the set of all sums $a+b$ with $a \in A \cup \{0\}$, $b \in B \cup \{0\}$. Note that here we allow to not choose any value in one or both sets by adding $\{0\}$. Often we consider the \emph{$t$-capped sumset} $A \plt B$, which is simply the restriction of $A \pl B$ to $\{0,\ldots,t\}$. 
Sumset computation is the most important primitive of our algorithm.
Easy reductions transform this problem to Boolean convolution and further to integer multiplication, which has well-known algorithms using FFT. This yields an algorithm for computing $A \plt B$ in time $\Oh(t \log t)$ (see \secref{preliminaries} for details). 

\paragraph{Color-Coding}

Color-coding is an algorithmic technique that was first used for the \textsc{$k$-Path} problem: Given a graph $G$, decide whether it contains a path of length $k$~\cite{alon1995color}. The idea is to randomly color the vertices of $G$ with $k$ colors, so that for a fixed path of length $k$ in $G$ with probability $1/k!$ it is colored $(1,2,\ldots,k)$, in which case we can find it by a simple dynamic programming algorithm on the layered graph obtained from keeping only the edges in $G$ from color class $i$ to $i+1$ (for each~$i$). Over $\Oh(k! \log n)$ repetitions we find a $k$-path with high probability, if one exists. 
Research on color-coding has led to various improvements and derandomizations of this technique~\cite{alon1995color,naor1995splitters,schmidt1990spatial}. The derandomizations also apply to our use of color-coding in this paper, however, other parts of our algorithm seem impossible to derandomize and we leave this as an open problem.

We make use of color-coding for finding all sums generated by small subsets. More precisely, given a \SubsetSum instance $(Z,t)$ and a threshold $k$, we compute a set $S \subseteq \SSB{Z}{t}$ containing any sum generated by a subset $Y \subseteq Z$ of size $|Y| \le k$ with constant probability. This can be boosted to any high probability by repeating and taking the union. The main trick is to randomly partition $Z = Z_1 \cup \ldots \cup Z_{k^2}$ by assigning to any element $z \in Z$ a color in $\{1,\ldots,k^2\}$ independently and uniformly at random. Consider the sumset $Z_1 \plt \ldots \plt Z_{k^2}$. Note that this set consists of sums generated by subsets of $Z$, in particular we never add up the same number twice since the $Z_i$ are disjoint. We say that the random partition \emph{splits} $Y$ if each $Z_i$ contains at most\footnote{In contrast to typical uses of color-coding, we want \emph{at most} one element instead of \emph{exactly} one element.} one element of~$Y$. In this case, the sum $\SUM(Y) := \sum_{y \in Y} y$ is contained in the sumset $Z_1 \plt \ldots \plt Z_{k^2}$. Indeed, by the definition of~$\plt$, in each position $i$ we can choose a number in $Z_i \cup \{0\}$, so if $Z_i$ contains exactly one element of~$Y$ then we can choose this element, while if $Z_i$ contains no element of $Y$ we can choose~0, to write $\SUM(Y)$ as a sum in the sumset $Z_1 \plt \ldots \plt Z_{k^2}$. It remains to argue that the random partition splits $Y$ with constant probability. The color-coding assigns each element in $Y$ a random color in $\{1,\ldots,k^2\}$, and thus we can view the partitioning, restricted to $Y$, as placing $k$ balls (the elements of $Y$) into $k^2$ bins (the subsets $Z_i$). This process is well understood, in particular the birthday paradox implies for our choice of $k^2$ bins that with constant probability some $Z_i$ contains more than one element of $Y$ -- which is a bad event in our situation. However, since the bound of the birthday paradox is tight, also with constant probability we have $|Z_i \cap Y| \le 1$ for all $i$, and thus the partitioning splits $Y$. Note that the running time for randomly partitioning and computing the sumset $Z_1 \plt \ldots \plt Z_{k^2}$ is $\tOh(n + k^2 t)$, which is near-linear in $n+t$ if $k$ is polylogarithmic in $n$ and~$t$.

\paragraph{Layer Splitting}
Any \SubsetSum instance $(Z,t)$ can be partitioned into $\log n$ \emph{layers} $Z_i \subseteq [t/2^i,t/2^{i-1}]$, plus a set $Z_0 \subseteq [0,t/n]$ that can be treated very similarly to layers and that we ignore here for simplicity. 
We obtain the desired sumset $\SSB{Z}{t}$ by combining the sumsets $\SSB{Z_i}{t}$ of the layers in a straight-forward way using sumset computations. 

It remains to compute $\SSB{Z_i}{t}$. By the definition of layers, any set $Y \subseteq Z_i$ summing to at most~$t$ has size $|Y| \le 2^i$. Thus, the color-coding algorithm with $k=2^i$ computes $\SSB{Z_i}{t}$, and it runs in time $\tOh(n + t k^2)$. To obtain an $\tOh(n+t)$ algorithm we need to eliminate the factor $k^2$.
We remove one factor $k$ by observing that all items in $Z_i$ are bounded by $\Oh(t/k)$, which allows us to implement the sumset computations in the color-coding algorithm more efficiently. The remaining factor $k$ stems from color-coding partitioning $Z_i$ into $k^2$ subsets. We remove this factor by a two-stage approach: In the first stage, we partition $Z_i$ into roughly $k = 2^i$ sets $Z_{i,j}$. This $k$ is too small to split $Y$ entirely, i.e., to have $|Z_{i,j} \cap Y| \le 1$ for all $j$ with high probability. Indeed, for this property we would need to partition $Z_i$ into $k^2$ sets. However, we still have  $|Z_{i,j} \cap Y| \le \Oh(\log n)$ with high probability. Hence, in the second stage we can run the color-coding algorithm with size bound $k' = \Oh(\log n)$ on each $Z_{i,j}$, and then combine their computed sumsets in a straight-forward way. This removes the factor\nobreakdash-$k$ overhead from partitioning into $k^2$ sets.
Carefully implementing these ideas yields time $\tOh(n+t)$.

\paragraph{Polynomial Space}
The polynomial space algorithm by Lokshtanov and Nederlof interprets a \SubsetSum algorithm as a circuit, where each gate performs the convolution or pointwise addition of two vectors of some fixed length $f(n,t)$. This circuit $C$ is transfered to the Fourier domain by replacing every convolution gate by pointwise multiplication. The new circuit $C'$ computes the Fourier transform of the output vector of $C$.
Using that the inverse Fourier transform can be written as a simple sum, we can evaluate an entry of $C$ by evaluating all entries of $C'$ one-by-one. Writing $g(n,t)$ for the number of gates in $C$, since all operations in $C'$ are pointwise, computing an entry of $C'$ can be done in $\tOh(g(n,t))$ arithmetic operations, storing $\tOh(g(n,t))$ numbers. Evaluating all entries of $C'$ one-by-one thus can be done in $\tOh(g(n,t)\cdot f(n,t))$ arithmetic operations, storing $\tOh(g(n,t))$ numbers.

One problem in the algorithm by Lokshtanov and Nederlof is that entries can become as large as $2^{\Omega(n)}$, so they need to work with $\tOh(n)$ bits of precision. This yields total time $\tOh(n \cdot g(n,t) f(n,t))$ and space $\tOh(n \cdot g(n,t))$. We work with a variant of their algorithm using modular arithmetic instead of complex numbers.
This allows us to work \emph{modulo a random prime $p$}, thus reducing the precision from $\tOh(n)$ to $\Oh(\log p)$ bits, and improving time and space by a factor $\tOh(n / \log p)$. A technical difficulty is that for using the Fourier transform the field $\Z_p$ has to contain a primitive $t$\nobreakdash-th root of unity. This is guaranteed by choosing $p$ in the arithmetic progression $1 + t \cdot \N$. However, to be able to choose $p$ at random from a sufficiently large ground set, we need to choose a threshold~$x$ such that there are many primes $p \le x$ in the arithmetic progression $1 + t \cdot \N$. This requires a quantitative version of Dirichlet's theorem. The best such result assumes ERH and allows us to choose $p \le x = n^{\Oh(1)}$, yielding an improvement factor $\tOh(n/\log p) = \tOh(n)$. 

The specific \SubsetSum circuit of Lokshtanov and Nederlof needs $g(n,t)=\Oh(n)$ gates and vector length $f(n,t) = \tOh(nt)$. A priori it is not easy to improve $f(n,t)$ to $\tOh(t)$. We show that the circuit induced by our new \SubsetSum algorithm allows us to set $f(n,t) = \tOh(t)$ and $g(n,t) = \tOh(n)$, thus improving the running time by another factor $\tOh(n)$.

\subsection{Further Related Work}


\SubsetSum has been studied extensively, see e.g.~\cite{kellererknapsack}. The best-known running time in terms of $n$ is $\Oh^*(2^{n/2})$~\cite{horowitz1974computing}. There is a large literature improving this running time for inputs that are in some sense ``structured'', see, e.g., \cite{austrin2015subset} and the references therein.
Using number-theoretic arguments, certain \emph{dense} cases of \SubsetSum are solvable in near-linear time, e.g., if $t \ll n^2$ then there is an $\tOh(n)$ algorithm~\cite{galil1991almost}. A \emph{$(1-\eps)$-approximation} for \SubsetSum yields a set with sum in $[(1-\varepsilon) t, t]$ if one exists. 
The best known approximation algorithm runs in time $\tOh(\min\{n/\eps, n + 1/\eps^2\})$~\cite{lawler1979fast,gens1980fast,kellerer2003efficient}.

\section{Preliminaries}
\label{sec:preliminaries}

All logarithms ($\log$) in this paper are base 2. 

\paragraph{Sums and Sumsets} For a set $S$ of integers we denote the sum of its elements by $\SUM(S) := \sum_{s \in S} s$. Given a set $Z$ of $n$ positive integers and target~$t$, the \SubsetSum problem asks whether there exists $Y \subseteq Z$ with $\SUM(Y) = t$. We often solve the more general problem of computing the set of all subset sums of $Z$ bounded by $t$, i.e.,
$$\SSB{Z}{t} := \{\SUM(Y) \mid Y \subseteq Z\} \cap \{0,\ldots,t\}. $$

We represent a set $S$ of integers in $\{0,\ldots,m\}$ by its characteristic vector of length $m+1$.
For sets $A,B$ of non-negative integers, we define their \emph{sumset}, in a slightly non-standard way, as the set of all sums of at most one element of $A$ and at most one element of $B$, i.e., 
$$A\pl B = \{a+b \mid a \in A \cup \{0\}, \, b \in B \cup \{0\}\}.$$
For any $t>0$, we define the \emph{$t$-capped sumset} as $A \plt B := (A \pl B) \cap \{0,\ldots,t\}$.

\paragraph{Sumset Computation}
For Boolean vectors $x,y$ of length $t$ we define their convolution as the Boolean vector $z$ of length $2t$ with $z_i = \bigvee_{j =1}^i x_j \wedge y_{i-j}$, where out-of-bounds values are interpreted as false. Observe that the convolution of the characteristic vectors of integer sets $A,B$ equals the characteristic vector of the sumset $A \pl B$ (if $0 \in A,B$). Thus, sumset computation can be reduced to Boolean convolution. 

A simple algorithm for Boolean convolution is to reduce to integer multiplication: Convert the Boolean vector $x$ to a number by interpreting true as $1$ and false as $0$ and padding with $\log t$ zeroes between any two bits of $x$. Do the same with $y$. Then from the product of the constructed numbers we can infer $z$, since a block of $1+\log t$ bits is identically 0 if and only if the corresponding bit of~$z$ is 0. Hence, if we can multiply two $t$-bit numbers in time $M(t)$, then Boolean convolution and thus sumset computation can be performed in time $\Oh(M(t \log t))$. The best known bound for $M(t)$ depends on the machine model. For simplicity, in this paper we work on the RAM model with cell size $w = \Theta(\log t)$, where typical operations on $\log(t)$-bit numbers can be performed in constant time. In this model it is known that $M(t) = \Oh(t)$, see, e.g.,~\cite{furer2014fast}\footnote{On other machine models $M(t)$ can be larger, e.g., on 2\nobreakdash-tape Turing machines or as circuits we have $M(t) \le t \log t \cdot 2^{\Oh(\log^* t)}$~\cite{furer2009faster,de2013fast}.}.
This yields the following:

\begin{proposition} \label{propo:sumsetcomputation}
  Given sets $A,B$ of non-negative integers, the $t$-capped sumset $A \plt B$ can be computed in time $\Oh(t \log t)$. 
\end{proposition}

%

\paragraph{Preprocessing Multisets} 
So far we assumed that the input $Z$ is a \emph{set}.  
Note that it also makes sense to allow $Z$ to be a \emph{multi-set}, and define a subset $Y \subseteq Z$ to be a multi-set where each $z \in Z$ has multiplicity in $Y$ at most its multiplicity in $Z$. Thus, any $z \in Z$ with multiplicity $m$ may be used at most $m$ times in any subset sum. 
E.L.~Lawler~\cite{lawler1979fast} showed that the general case of multi-sets can be reduced to multi-sets with multiplicities bounded by 2. To this end, for any $z \in Z$ with multiplicity $2k+1$ we reduce its multiplicity to 1 and add $k$ times the integer $2z$ to $Z$ (i.e. increase the multiplicity of $2z$ by $k$). Observe that the resulting set generates the same subset sums, since we have $\{i \cdot z \mid 0 \le i \le 2k+1\} = \{i \cdot z + j \cdot 2z \mid 0 \le i \le 1 \text{ and } 0 \le j \le k\}$. Similarly, if $z$ has multiplicity $2k+2$ we reduce its multiplicity to 2 and add $k$ times $2z$. Repeating this step for all $z \in Z$ (in increasing order) yields an equivalent \SubsetSum instance where all multiplicities are bounded by 2, so we obtain the following. 
For a multi-set $Z$, by the \emph{size} $|Z|$ we denote the sum of all multiplicities of elements in $Z$.

\begin{proposition} \label{propo:lawler}
  Given a \SubsetSum instance $(Z,t)$, where $Z$ is a multi-set of size $n$, we can in time $\Oh(n+t)$ compute an equivalent instance $(Z',t)$ where $Z'$ is a multi-set with multiplicities bounded by 2 and $|Z'| \le \min\{n,2t\}$.
\end{proposition}

Koiliaris and Xu~\cite{KoiliarisX15} extended this preprocessing as follows. Write the multi-set $Z'$ given by \proporef{lawler} as a union of two \emph{sets} $Z_1,Z_2 \subseteq \{0,\ldots,t\}$. Compute $\SSB{Z_1}{t}$ and $\SSB{Z_2}{t}$. Then we obtain $\SSB{Z'}{t}$ as $\SSB{Z_1}{t} \plt \SSB{Z_1}{t}$. This yields a reduction to sets, running in time $\Oh(n + t \log t)$. A disadvantage is that we have to compute the whole sets $\SSB{Z_{1}}{t},\SSB{Z_{2}}{t}$ instead of just solving a decision problem, but this is irrelevant for the algorithms presented in this paper. 

By running this reduction as a preprocessing of all our algorithms, throughout the paper we can assume that the input $Z \subseteq \{0,\ldots,t\}$ is a set.

\section{Near-Linear Time Algorithm}
\label{sec:simple}

Our goal is to compute, given a set $Z$ of $n$ positive integers and target~$t$, the set $\SSB{Z}{t} := \{\SUM(Y) \mid Y \subseteq Z\} \cap \{0,\ldots,t\}$, since checking whether $t \in \SSB{Z}{t}$ decides the given \SubsetSum instance. In this section, we design a simple algorithm that solves \SubsetSum in time $\tOh(n + t)$, proving \thmref{main}. 
As discussed in \secref{techniques}, our algorithm consists of two parts: We first show how to find sums generated by small subsets using color-coding in \secref{colorcodingone}, and then use a two-stage approach on layers for finding all subset sums in \secref{layers}.

\subsection{Color-Coding: An algorithm for small solution size}
\label{sec:colorcodingone}

We describe an algorithm $\ColorCoding$ (see below) for solving \SubsetSum if the solution size is small, i.e., an algorithm that finds all sums $\SUM(Y) \le t$ generated by sets $Y \subseteq Z$ of size $|Y| \le k$, for some given (small)~$k$. 
We \emph{randomly partition} $Z = Z_1 \cup \ldots \cup Z_{k^2}$, i.e., we assign any $z \in Z$ to a set $Z_i$ where $i$ is chosen independently and uniformly at random in $\{1,\ldots,k^2\}$. We say that this random partition \emph{splits} $Y$ if $|Y \cap Z_i| \le 1$ holds for all $1 \le i \le k^2$. If this happens, then the sumset $Z_1 \plt \ldots \plt Z_{k^2}$ contains $\SUM(Y)$. Indeed, by definition of $\pl$ in each position $i$ we can choose a number in $Z_i \cup \{0\}$, so for $|Y \cap Z_i| = 1$ we can choose the unique number in the set $Y \cap Z_i$, while for $|Y \cap Z_i| = 0$ we can choose~0, to generate $\SUM(Y)$ as a sum in $Z_1 \plt \ldots \plt Z_{k^2}$. Also note that the sumset $Z_1 \plt \ldots \plt Z_{k^2}$ only contains valid sums in $\SSB{Z}{t}$, since no $z \in Z$ may be used twice. 

Repeating this procedure sufficiently often with fresh randomness, and taking the union over all computed sumsets $Z_1 \plt \ldots \plt Z_{k^2}$, yields a set $S \subseteq \SSB{Z}{t}$ containing any $\SUM(Y) \le t$ with $Y \subseteq Z$ and $|Y| \le k$ with probability at least $1-\delta$.
In fact, $\Oh(\log 1/\delta)$ repetitions are sufficient, since a random partition splits $Y$ with constant probability, which follows from the birthday paradox, or more precisely the tightness of the birthday paradox bound. We briefly prove this standard claim for completeness. 
Since $Z \subseteq Z_1 \plt \ldots \plt Z_{k^2}$, we can assume $k \ge 2$.
For any $Y \subseteq Z$ with $|Y| \le k$, the probability of the random partition splitting $Y$ is the same as the probability of $|Y|$ balls falling into $|Y|$ different bins, when throwing $|Y|$ balls into $k^2$ bins. This is equivalent to the second ball falling into a different bin than the first one, the third ball falling into a different bin than the first two, and so on, which has probability
  \begin{align*}
  &\frac{k^2-1}{k^2} \cdot \frac{k^2-2}{k^2} \ldots \frac{k^2-(|Y|-1)}{k^2}  \IfSODA{\\
  &}\ge \bigg(\frac{k^2-(|Y|-1)}{k^2}\bigg)^{|Y|} \ge \Big(1-\frac 1k\Big)^{k} \ge \Big(\frac 12\Big)^{2} = \frac 14.
  \end{align*}
Hence, $r := \lceil \log_{4/3}(1/\delta) \rceil$ repetitions yield the desired success probability of $1 - (1-1/4)^{r} \ge 1-\delta$. This finishes the analysis of $\ColorCoding$ and proves the following lemma. For the running time, note that we perform $\Oh(\log 1/\delta)$ repetitions of computing $k^2$ sumsets, each taking time $\Oh(t \log t)$.

\begin{lem} \label{lem:colorcoding}
  $\ColorCoding(Z,t,k,\delta)$ computes in time $\Oh(t k^2 \log t \log (1/\delta))$ a set $S \subseteq \SSB{Z}{t}$ such that for any $Y \subseteq Z$ with $|Y| \le k$ and $\SUM(Y) \le t$ we have $\SUM(Y) \in S$ with probability $\ge 1-\delta$.
\end{lem}

\begin{algorithm}
\caption{$\ColorCoding(Z,t,k,\delta)$: Given a set $Z$ of positive integers, target~$t$, size bound $k\ge 1$ and error probability $\delta > 0$, we solve \SubsetSum with solution size at most $k$}\label{alg:colorcoding}
\begin{algorithmic}[1]
\For {$j=1,\ldots,\lceil \log_{4/3}(1/\delta) \rceil$}
  \State randomly partition $Z = Z_1 \cup \ldots \cup Z_{k^2}$
  \State $S_j := Z_1 \plt \ldots \plt Z_{k^2}$
\EndFor
\State \Return $\bigcup_j S_{j}$
\end{algorithmic}
\end{algorithm}

Standard techniques allow to derandomize this algorithm by iterating over a (deterministic) family of partitions of $Z$ that is guaranteed to contain a partition splitting $Y$. This comes at the cost of an increased polynomial factor in~$k$, e.g., \cite[Lemma 2]{naor1995splitters} gives a factor $\Oh(k^6 \log k)$. Alternatively, Koiliaris and Xu provide a very different algorithm running in deterministic time $\Oh(t k^2 \log (tk) \log n)$~\cite[Lemma~2.12]{KoiliarisX15}.

\subsection{Layer Splitting}
\label{sec:layers}

Let $(Z,t)$ be a \SubsetSum instance and $|Z|=n$. For $\ell \ge 1$, we call $(Z,t)$ an \emph{$\ell$\nobreakdash-layer} instance if 
$$ Z \subseteq [t/\ell,2t/\ell] \qquad\qquad \text{or} \qquad\qquad Z \subseteq [0,2t/\ell] \; \text{ and } \; \ell \ge n. $$
In both cases we have $Z \subseteq [0,2t/\ell]$. Moreover, observe that any $Y \subseteq Z$ summing to at most~$t$ has size $|Y| \le \ell$. In the second case, this holds since $|Y| \le |Z| = n$. 
Thus, the $\ColorCoding$ algorithm from the last section with size bound $k = \ell$ solves $\ell$-layer instances. In this section, we will show that the running time of $\ColorCoding$ can be improved for layers, thereby essentially removing the quadratic dependence on $\ell$ entirely. 

However, before discussing how to solve \SubsetSum on layers, we show that any given instance $(Z,t)$ can be split into $\Oh(\log n)$ layers, see \algref{fastersubsetsum} below. 
We simply split set $Z$ at $t/2^i$ for $i=1,\ldots,\lceil \log n \rceil -1$; this yields $\Oh(\log n)$ layers $Z_1,\ldots,Z_{\lceil \log n \rceil}$.  On each layer we then run the algorithm $\ColorCodingLayer$ presented below, and we combine the resulting sumsets $S_i$ in a straight-forward way. The error probabilities of the calls to $\ColorCodingLayer$ are chosen sufficiently small so that they sum up to at most~$\delta$. Correctness and the running time bound $\Oh(t \log t \log^3(n/\delta) \log n)$ immediate follow from \lemref{colorcodinglayer} below. This yields algorithm $\FasterSubsetSum$ below that proves \thmref{main}.

\begin{algorithm}
\caption{$\FasterSubsetSum(Z,t,\delta)$: Returns a set $S \subseteq \SSB{Z}{t}$ containing any $s \in \SSB{Z}{t}$ with probability at least $1-\delta$, and runs in time $\Oh(t \log t \log^3(n/\delta) \log n)$}
\label{alg:fastersubsetsum}
\begin{algorithmic}[1] 
\State split $Z$ into $Z_i := Z \cap (t/2^{i},t/2^{i-1}]$ for $i=1,\ldots,\lceil \log n \rceil - 1$, and $Z_{\lceil \log n \rceil} := Z \cap [0,t/2^{\lceil \log n \rceil -1}]$ 
\State $S = \emptyset$
\For {$i=1,\ldots, \lceil \log n \rceil$}
  \State $S_i := \ColorCodingLayer(Z_i,t,2^{i},\delta / \lceil \log n \rceil)$
  \State $S := S \plt S_i$
\EndFor
\State \Return $S$
\end{algorithmic}
\end{algorithm}


It remains to design a fast algorithm given an $\ell$-layer instance $(Z,t)$ and error probability~$\delta$.
Let $m := \ell/\log(\ell/\delta)$ rounded up to the next power of~2.
We randomly partition $Z$ into  subsets $Z_1,\ldots,Z_m$. For each $Z_j$ we run $\ColorCoding$ with size bound $k = 6 \log(\ell/\delta)$, target $12 \log(\ell/\delta) t/\ell$, and error probability $\delta/\ell$, yielding a set~$S_j$. 
We combine the sets $S_1, \ldots, S_m$ in a natural, binary-tree-like way by computing $S_1 \pl S_2, S_3 \pl S_4, \ldots, S_{m-1} \pl S_m$ in the first round, $S_1 \pl S_2 \pl S_3 \pl S_4, \ldots, S_{m-3} \pl S_{m-2} \pl S_{m-1} \pl S_m$ in the second round, and so on, until we reach the set $S_1 \pl \ldots \pl S_m$. Note that in the $h$-th round of this procedure we combine $2^h$ sets $S_j$, initially containing integers bounded by $12 \log(\ell/\delta) t/\ell$. Thus, we may use $\pl[2^h \cdot 12 \log(\ell/\delta) t/\ell]$ in the $h$-th round. This explains the following algorithm.

\begin{algorithm}
\caption{$\ColorCodingLayer(Z,t,\ell,\delta)$: See \lemref{colorcodinglayer} for guarantees
}\label{alg:colorcodinglayer}
\begin{algorithmic}[1]
\State \textbf{if} $\ell < \log(\ell/\delta)$ \textbf{then return} $\ColorCoding(Z,t,\ell,\delta)$
\State $m := \ell/\log(\ell/\delta)$ rounded up to the next power of 2
\State randomly partition $Z = Z_1 \cup \ldots \cup Z_{m}$
\State $\gamma := 6 \log(\ell/\delta)$
\For {$j=1,\ldots,m$}
  \State $S_j := \ColorCoding(Z_j,2 \gamma t/\ell, \gamma, \delta/\ell)$
\EndFor
\For {$h=1,\ldots,\log m$} \IfNotSODA{\Comment{combine the sumsets $S_j$ in a binary-tree-like way}}
  \For {$j=1,\ldots,m / 2^{h}$}
    \State $S_j := S_{2j-1} \pl[2^h \cdot 2 \gamma t/\ell] S_{2j}$
  \EndFor
\EndFor
\State \Return $S_{1} \cap \{0,\ldots,t\}$
\end{algorithmic}
\end{algorithm}

\begin{lem} \label{lem:colorcodinglayer}
  For an $\ell$-layer instance $(Z,t)$ and $\delta \in (0,1/4]$, the method $\ColorCodingLayer(Z,t,\ell,\delta)$ computes in time $\Oh(t \log t \log^3(\ell/\delta))$ a set $S \subseteq \SSB{Z}{t}$ containing any $s \in \SSB{Z}{t}$ with probability at least $1-\delta$.
\end{lem}
\begin{proof}
The case $\ell < \log(\ell/\delta)$ follows from \lemref{colorcoding}.
The inclusion $S \subseteq \SSB{Z}{t}$ also follows from \lemref{colorcoding} and since we only compute sumsets over partitionings, using the fact $\big(\SSB{Z_1}{t_1} \pl[t'] \SSB{Z_2}{t_2}\big) \cap \{0,\ldots,t\} \subseteq \SSB{Z}{t}$ for a partitioning $Z=Z_1 \cup Z_2$ and any $t,t',t_1,t_2 \ge 1$.

For the error probability, fix a subset $Y \subseteq Z$ with $\SUM(Y) \le t$, and let $Y_j := Y \cap Z_j$ for $1 \le j \le m$. A crucial property is that with sufficiently high probability the size $|Y_j|$ is smaller than $6 \log(\ell/\delta)$, thus allowing us to run $\ColorCoding$ with size bound $k = 6 \log(\ell/\delta)$.

\begin{claim}
  We have $\Pr[|Y_j| \ge 6 \log(\ell/\delta)] \le \delta/\ell$. 
\end{claim}
\begin{proof}
  Note that $|Y_j|$ is distributed as the sum of $|Y|$ independent Bernoulli random variables with success probability $1/m$. In particular, $\mu := \Ex[|Y_j|] = |Y| / m$. A standard Chernoff bound yields that $\Pr[|Y_j| \ge \lambda] \le 2^{-\lambda}$ for any $\lambda \ge 2 e \mu$.
  Recall that $(Z,t)$ is an $\ell$-layer instance, and thus $Y$ has size at most~$\ell$. This allows us to bound $\mu = |Y| / m \le \ell/m \le \log (\ell/\delta)$, by definition of $m$. The concentration inequality thus holds for $\lambda = 6 \log(\ell/\delta)$, and we obtain $\Pr[|Y_j| \ge 6 \log(\ell/\delta)] \le \delta/\ell$.  \IfSODA{\qed}
\end{proof}
By the above claim, we may assume that $|Y_j| \le 6 \log(\ell/\delta)$ holds for each $1 \le j \le m$; this happens with probability at least $1-m \cdot \delta/\ell$. 
Since $Z \subseteq [0,2t/\ell]$, any subset of $Z_j$ of size at most $6 \log(\ell/\delta)$ has sum at most $12 \log(\ell/\delta) t/\ell$. It follows that the call $\ColorCoding(Z_j, 12 \log(\ell/\delta) \cdot t/\ell, 6 \log(\ell/\delta), \delta/\ell)$ finds $\SUM(Y_j)$ with probability at least $1-\delta/\ell$. Assume that this event holds for each $1 \le j \le m$; this happens with probability at least $1-m \cdot \delta/\ell$. Then $S_j$ contains $\SUM(Y_j)$, and the tree-like sumset computation indeed yields a set containing $\SUM(Y_1) + \ldots + \SUM(Y_m) = \SUM(Y)$. The total error probability is $2 m \delta / \ell$. Since $\ell \ge 1$ and $\delta \le 1/4$ we have $\log(\ell/\delta) \ge 2$ and obtain $m \le \ell/2$. Hence, the total error probability is bounded by $\delta$. 

Regarding the running time, observe that by \lemref{colorcoding} each call to $\ColorCoding$ takes time $\Oh( t/\ell \cdot \log^4(\ell/\delta) \log t)$. Since there are $m = \Theta(\ell / \log(\ell/\delta))$ calls to $\ColorCoding$, we obtain the claimed time $\Oh(t \log^3(\ell/\delta) \log t)$. The remaining time for combining the sets $S_1,\ldots,S_m$ is 
$$\Oh\bigg(\sum_{h=1}^{\log m} \frac{m}{2^{h}} \cdot 2^h \log(\ell/\delta) t/\ell \cdot \log t\bigg) = \Oh(t \log t \log m), $$
which is dominated by the total time for calling $\ColorCoding$.
\IfSODA{\qed}
\end{proof}

\section{Polynomial Space Algorithm}
\label{sec:polynomial}

\subsection{Setup}
\label{sec:polynomialsetup}

We first give the setup of Lokshtanov and Nederlof~\cite{lokshtanov2010saving} and then present a version of their main result using modular arithmetic instead of complex numbers.

\paragraph{Circuits}
For a set $F$ and binary operators $O_1$, $O_2$ on $F$, a circuit $C$ over $(F; O_1, O_2)$ is a directed acyclic graph $D = (N, A)$ with parallel arcs, such that every node of $D$ is either a constant gate (indegree 0), $O_1$ gate (indegree 2), or $O_2$ gate (indegree 2). For an indegree 2 node, its two in-neighbours are its children. The output of a constant gate is the element it is labeled with. The output of an $O_i$ gate is the result of performing $O_i$ on the output of its two children. The output of $C$, denoted by $\outC$, is the output of a specific gate $c$ marked as the output gate. The size of $C$ is the size of the underlying graph (number of vertices plus number of edges). With some abuse of notation we will denote a gate of $C$ and the output of that gate by the same symbol. A \emph{probabilistic} circuit is a distribution $\C$ over circuits over $(F; O_1, O_2)$. 

\paragraph{Vector Operations}
Fix a ring $R$ with operations $+,\cdot$. For a vector $\va \in R^t$ we write $\va[i]$, $0 \le i < t$, for its $i$-th entry. We call $\va$ a \emph{singleton} if it has at most one non-zero entry.
We write $\boxplus$ for pointwise addition and $\boxdot$ for pointwise multiplication, i.e.,
for vectors $\va,\vb \in R^t$ we have $(\va \boxplus \vb)[i] = \va[i] + \vb[i]$ and $(\va \boxdot \vb)[i] = \va[i] \cdot \vb[i]$. We write $\boxtimes$ for convolution, i.e., we have $(\va \boxtimes \vb)[i] = \sum_{j=0}^i \va[j] \cdot \vb[i-j]$. 
The \emph{length} of a vector $\va$ is the maximal index of any non-zero entry of $\va$, i.e., $\len(\va) := \max\{0 \le i < t \mid \va[i] \ne 0\}$ or 0, if $\va$ is the all-zeroes-vector. 
We say that the convolution $\va \boxtimes \vb$ \emph{overflows} if $\len(\va) + \len(\vb) \ge t$. If $\va \boxtimes \vb$ does not overflow and $\va,\vb$ are the coefficients of polynomials $P(x),Q(x)$, respectively, then $\va \boxtimes \vb$ are the coefficients of the polynomial $P(x) \cdot Q(x)$. Throughout the paper we will always work with non-overflowing convolutions. 
 Finally, we also consider the sum of all entries of $\va$, i.e., $\summ(\va) := \sum_{i=0}^{t-1} \va[i]$.
  
  \begin{lem} \label{lem:lenmax}
    Let $\va_1,\ldots,\va_k \in R^{t}$. We have
    \begin{enumerate}[label=(\arabic*)]
      \item $\len(\va_1 \boxplus \ldots \boxplus \va_k) = \max_i \len(\va_i)$,
      \item $\summ(\va_1 \boxplus \ldots \boxplus \va_k) = \sum_i \summ(\va_i)$,
      \item if there is no overflow then $\len(\va_1 \boxtimes \ldots \boxtimes \va_k) = \sum_i \len(\va_i)$, and
      \item if there is no overflow then $\summ(\va_1 \boxtimes \ldots \boxtimes \va_k) = \prod_i \summ(\va_i)$.
    \end{enumerate}
  \end{lem}
  \begin{proof}
    (1), (2), and (3) are immediate. For (4), note that any product $\va_1[j_1]\cdot \ldots \cdot \va_k[j_k]$, with $0 \le j_i \le \len(\va_i)$, contributes to exactly one entry of $\vb := \va_1 \boxtimes \ldots \boxtimes \va_k$, and each entry of $\vb$ can be written as a sum of such products.    
    \IfSODA{\qed}
  \end{proof}

\paragraph{Discrete Fourer Transform}
Fix a prime $p$ and $t \ge 1$. 
Let $\omega \in \Z_p$ be a primitive $t$-th root of unity, i.e., $\omega^t = 1$ and $\omega^{k} \ne 1$ for $1 \le k < t$.
The $t$-point Discrete Fourier Transform (DFT) over modular arithmetic in $\Z_p$ is the linear function $\F$ mapping a vector $\va \in \Z_p^t$ to $\F(\va) \in \Z_p^t$ with
$$ \F(\va)[i] = \sum_{j=0}^{t-1} \omega^{i j} \va[j]. $$
In other words, if $\va$ is the vector of coefficients of the polynomial $P(x) = \sum_{i=0}^{t-1} \va[i] x^i$ then $\F$ evaluates $P(x)$ at all powers of $\omega$, i.e., $\F(a) = (P(1),P(\omega),P(\omega^2),\ldots,P(\omega^{t-1}))$.
The inverse function $\F^{-1}$  is given by
$$ \F^{-1}(\va)[j] = \frac 1t \sum_{i=0}^{t-1} \omega^{-i j} \va[i]. $$

That $\F^{-1}$ is indeed the inverse operation of $\F$ is a consequence of the identity $\sum_{i=0}^{t-1} \omega^{-ij} \omega^{ik} = t \cdot [j=k]$, which follows from $\omega$ being a primitive $t$-th root of unity. Here and in the remainder we use Iverson's bracket notation, i.e., $[B]$ is 1 if $B$ is true, and 0 otherwise.
The most important property of DFT for our purposes is the Convolution Theorem (see, e.g.,~\cite{CLRSthird}), stating that for any $\va,\vb \in \Z_p^t$ such that $\va \boxtimes \vb$ does not overflow we have
$$ \F(\va \boxtimes \vb) = \F(\va) \boxdot \F(\vb).$$

\paragraph{Modular Variant of Lokshtanov-Nederlof}
Lokshtanov and Nederlof~\cite{lokshtanov2010saving} mention without proof that a version of their main result holds for modular arithmetic. For completeness, we prove such a result here. 

\begin{thm} \label{thm:lokned}
  Let $p$ be prime, $t \ge 1$, and suppose that $\Z_p^t$ contains a $t$-th root of unity $\omega$. Let $C$ be a circuit over $(\Z_p^t,\boxplus,\boxtimes)$ with only singleton constants. Suppose that no convolution gate overflows. Then given $p,t,\omega$, and $0 \le x < t$ we can compute $\outC[x]$ in time $\tOh(|C| t \log p)$ and space $\Oh(|C| \log p)$. 
\end{thm}

We discuss how to find an appropriate root of unity~$\omega$ later in \lemref{findingomega}.
We remark that Lokshtanov and Nederlof have the depth of $C$ as an additional factor in their time and space bounds, which seems to be due to their choice of the complex Fourier transform and does not appear in the modular version (this is a lower-order improvement).

\begin{proof}
The proof by Lokshtanov and Nederlof works almost verbatim.
From the circuit $C$ we construct another circuit $C'$ over $(\Z_p^t,\boxplus,\boxdot)$ with the same directed graph as $C$, but with different gates. Each constant gate $\va \in C$ is replaced with the constant gate $\va' = \F(\va)$. Each convolution gate $\boxtimes$ is replaced with a pointwise multiplication gate $\boxdot$. Then an easy inductive argument shows that for any gate $\va \in C$ and its corresponding gate $\va' \in C'$ we have $\va' = \F(\va)$. Indeed, if $\va$ is a constant gate then the claim holds by definition. For the inductive step, consider let $\vc{b},\vc{c}$ be the children of $\va$. By the inductive hypothesis we have $\vc{b}' = \F(\vc{b})$ and $\vc{c}' = \F(\vc{c})$. If $\va$ is an addition gate, then we conclude $\va' = \vc{b}' \boxplus \vc{c}' = \F(\vc{b}) \boxplus \F(\vc{c}) = \F(\vc{b} \boxplus \vc{c}) = \F(\va)$. If $\va$ is a convolution gate, then the Convolution Theorem implies $\va' = \vc{b}' \boxdot \vc{c}' = \F(\vc{b}) \boxdot \F(\vc{c}) = \F( \vc{b} \boxtimes \vc{c}) = \F(\va)$.

For ease of notation we write $\vf := \outC$.
We now use $C'$ to compute $\vf[x]$. Plugging the definition of $\F^{-1}$ into the identity $\vf = \F^{-1}(\F(\vf))$ yields
$$ \vf[x] = \frac 1t \sum_{i=0}^{t-1} \omega^{-i x} (\F(\vf))[i]. $$
Note that $\omega^{s}$, for $s \le \poly(t)$, can be computed with $\Oh(\log t)$ arithmetic operations in $\Z_p$ by repeated squaring. Thus, in order to compute $\vf[x]$ it suffices to compute $(\F(\vf))[0],\ldots,(\F(\vf))[t-1]$. 

For computing $(\F(\vf))[j]$, we convert circuit $C'$ to a circuit $C'_j$ over $(\Z_p,+,\cdot)$ with the same directed graph as $C'$. Each constant gate $\va' \in C'$ is replaced by $\va'[j]$. Each  pointwise addition gate $\boxplus$ is replaced with $+$ and each pointwise multiplication gate $\boxdot$ is replaced with $\cdot$. It follows immediately that $C'_j$ computes $(\F(\vf))[j]$. 

Consider a (singleton) constant gate $\va \in C$ with $\va[k] = v$ and $\va[\ell] = 0$ for all $\ell \ne k$. By definition of $\F$ we have $(\F(\va))[j] = \omega^{j k} v$. Thus, any constant gate in $C'_j$ can be computed with $\Oh(\log t)$ arithmetic operations in $\Z_p$. Computing the remaining gates of $C'_j$ takes $|C'| = |C|$ arithmetic operations in $\Z_p$. As arithmetic operations in $\Z_p$ can be performed in time $\tOh(\log p)$, we can compute any value $(\F(\va))[j]$ in time $\tOh(|C| \log p \log t)$ and space $\Oh(|C| \log p)$. Summing over all $j=0,\ldots,t-1$ yields time $\tOh(|C| t \log p \log t) = \tOh(|C| t \log p)$ and space $\Oh(|C| \log p)$.
\IfSODA{\qed}
\end{proof}

\subsection{FasterSubsetSum as a Circuit}
\label{sec:polynomialfsscircuit}

In the remainder of the paper, set the error probability $\delta := 1/n$.
We now describe how to convert the algorithm $\FasterSubsetSum$ into a probabilistic circuit over $(\N^{\tOh(t)};\boxplus,\boxtimes)$. This yields:

\begin{lem} \label{lem:FSScircuit}
  Given a \SubsetSum instance $(Z,t)$ with $|Z| = n$, we can construct a probabilistic circuit $\C$ over $(\N^{t'};\boxplus,\boxtimes)$, where $t' = \Oh(t \log^3 n)$ is a power of two, satisfying
  \begin{enumerate}[label=(\arabic*)]
    \item for any $0 \le s \le t$, $s \not\in \SSB{Z}{t}$ we have $\Pr_{C \in \C}[\outC[s] = 0] = 1$, and 
    \item for any $0 \le s \le t$, $s \in \SSB{Z}{t}$ we have $\Pr_{C \in \C}[\outC[s] = 0] \le 1/n$. 
  \end{enumerate}
  Any circuit $C \in \C$ has size $|C| = \tOh(n)$, uses only singleton constants, has no overflowing convolution gates, and any gate $\va \in C$ satisfies $\summ(\va) \le 2^{\tOh(n)}$. Sampling a circuit $C$ from $\C$ can be performed in time $\tOh(n)$. 
\end{lem}
\begin{proof}
  We carefully inspect the algorithm $\FasterSubsetSum$, transforming any $\cup$-operation into a $\boxplus$-gate and any $\plt$-operation into a $\boxtimes$-gate. The details are as follows.

  Set $t' := 300 \log^3(2n^2) t$, rounded up to a power of two.
  For any vector $\va \in \N^{t'}$ we say that $\va$ \emph{represents} the set $\{ 0 \le i < t' \mid \va[i] > 0\}$. Note that if $\va,\vb$ represent sets $A,B$ then $\va \boxplus \vb$ represents $A \cup B$. Moreover, if $\va \boxtimes \vb$ does not overflow then $\va \boxtimes \vb$ represents $A \pl[t'] B$, which in this case equals $A \pl B$. For any $0 \le i < t'$ we denote by $\ve_i$ the singleton vector containing a 1 at position $i$ and a 0 at all other entries. Note that $\ve_i$ represents~$\{i\}$. 
  
  We now use this representation to convert algorithm $\FasterSubsetSum$ to a probabilistic circuit over $(\N^{t'};\boxplus,\boxtimes)$. Fix the randomness of $\FasterSubsetSum$, i.e., fix all random partitions performed in $\ColorCoding$ and $\ColorCodingLayer$. All remaining basic operations in $\FasterSubsetSum$, $\ColorCoding$, and $\ColorCodingLayer$ are (1) forming a subset $Z' \subseteq Z$, (2) taking the union of two sets, and (3) computing the sumset of two sets.
  Hence, we can view $\FasterSubsetSum$ as a circuit where each node computes a set. Each leaf is directly assigned a subset $Z' \subseteq Z$. Each inner node either computes the union $\cup$ or the sumset $\pl[t'']$, for some~$t''$, of its children. 
  
  Now we replace each leaf, corresponding to $Z' \subseteq Z$, by the circuit computing $\bigboxplus_{z \in Z'} \ve_z$, which is a vector representing $Z'$. Note that this circuit has size $|Z'|$ and uses only singleton constants. Moreover, we replace each $\cup$-gate by $\boxplus$ and each $\pl[t'']$-gate by $\boxtimes$. This yields a circuit $C$ over $(\N^{t'};\boxplus,\boxtimes)$. Over the randomness of $\FasterSubsetSum$ used for the random partitionings, we obtain a probabilistic circuit $\C$.
  
  We show that $t'$ is chosen sufficiently large so that circuit $C$ has no overflowing convolution gates. To this end, we use \lemref{lenmax} to bound the length $\len(\va) = \max\{0 \le i < t' \mid \va[i] \ne 0\}$ of any gate $\va$.
  For the output of $\ColorCoding(Z,t,k,\delta)$, as well as any of its inner gates, we can bound the length by $k^2 \cdot \max Z$, since we compute a union (which does not increase the length) over a $k^2$-wise sumset computation (multiplying the length by $k^2$) over subsets of $Z$ (having length at most $\max Z$). For $\ColorCodingLayer(Z,t,\ell,\delta)$ we can bound the length by $\ell \cdot 36 \log^2(\ell/\delta) \max Z$, since we compute an $m$-wise sumset, for some $m \le \ell$, over the output of $\ColorCoding$ with $k = 6 \log(\ell/\delta)$. Finally, for $\FasterSubsetSum(Z,t,\delta)$ we perform a sumset computation over layers $i=1,\ldots, \lceil \log n \rceil$. For the $i$-th layer we run $\ColorCodingLayer$ on $\ell = 2^i$ and $Z_i$ with $\max Z_i \le t/2^{i-1}$, which by the above bound yields length at most $2^i \cdot 36 \log^2(2^i/\delta) t/2^{i-1} \le 72 \log^2(2n/\delta) t$. Summing over all $i = 1,\ldots,\lceil \log n \rceil$ yields the desired bound on the length of $\lceil \log n \rceil \cdot 72 \log^2(2n/\delta) t \le 144 \log^3(2n/\delta) t < t'/2$, since we set $\delta = 1/n$. 
  This bound holds at all gates of any circuit $C \in \C$. In particular, this shows that no convolution gate overflows, since $t'$ is chosen sufficiently large.
  
  With this property it is immediate that $C$ computes a vector $\outC$ representing a set $S'$ such that for $S := S' \cap \{0,\ldots,t\}$ we have $\FasterSubsetSum(Z,t,1/n) \subseteq S \subseteq \SSB{Z}{t}$. Indeed, by the same reasoning as for $\FasterSubsetSum$ we argue that the resulting circuit only performs set partitionings and sumset computations and thus only computes valid subset sums, proving $S' \subseteq \SSB{Z}{t'}$ and thus $S \subseteq \SSB{Z}{t}$. For the other direction, note that we drop the index from a $\pl[t'']$-gate when converting it to a $\boxtimes$-gate. Thus, the circuit $C$ does not compute the same set $S$ as $\FasterSubsetSum(Z,t,1/n)$, but $S$ may be a proper superset. However, this change cannot remove elements, so we obtain the inclusion $\FasterSubsetSum(Z,t,1/n) \subseteq S$. 
  
  From this, correctness follows immediately. Indeed, for $s \in \{0,\ldots,t\}$, $s \not \in \SSB{Z}{t}$ we obtain $s \not\in S$ and thus $\outC[s] = 0$. Also, for $s \in \SSB{Z}{t}$ we have $s \in \FasterSubsetSum(Z,t,1/n)$ with probability at least $1-1/n$ and thus $s \in S$ with probability at least $1-1/n$.
  
  Since we almost always simply partition the current set $Z$, except for $\ColorCoding$ where we take a union over $\Oh(\log(n))$ partitionings, each item in $Z$ appears in $\Oh(\log(n))$ leafs of the circuit. It is also easy to see that the depth of the circuit is $\Oh(\log(n))$. This allows us to bound $|C| \le \tOh(n)$, and it is easy to see that $C \in \C$ can be sampled in the same time bound.
  
  Finally, we bound the sum of entries $\summ(\va) = \sum_{i=0}^{t'-1} \va[i]$ for any gate $\va \in C$ with $C \in \C$, using \lemref{lenmax}.(4). For a subset $Z' \subset Z$ we implement a circuit representing $Z'$ with sum of entries equal to $|Z'| \le n$. For $\ColorCoding(Z,t,k,\delta)$ we have sum of entries $\Oh(\log(1/\delta) n^{k^2})$, since we take the union over $\Oh(\log(1/\delta))$ rounds (multiplying the sum of entries by $\Oh(\log(1/\delta))$) over $k^2$-wise sumset computations (raising the sum of entries to the $k^2$\nobreakdash-th power) over subsets of $Z$, which have sum of entries at most~$n$.  For $\ColorCodingLayer(Z,t,\ell,\delta)$ we can bound the sum of entries by $n^{\Oh(\log^2(\ell/\delta) \cdot \ell)}$, since we compute an $m$-wise sumset, for some $m \le \ell$, over the output of $\ColorCoding$ with $k = \Oh(\log(\ell/\delta))$. Finally, for $\FasterSubsetSum(Z,t,1/n)$ the sum of entries is bounded by $n^{\Oh(\log n \cdot \log^2(n) \cdot n)} = 2^{\Oh(\log^4(n) \cdot n)} = 2^{\tOh(n)}$, since we perform a sumset computation over $\Oh(\log n)$ layers (which raises the sum of entries to the $\Oh(\log n)$-th power) over the result of $\ColorCodingLayer$ on $\ell \le 2n$. 
  \IfSODA{\qed}
\end{proof}

Note that the bound $\summ(\va) \le 2^{\tOh(n)}$ also bounds the maximal entry of any gate $\va$ in any circuit $C \in \C$. Thus, we may replace $\N$ by $\Z_p$ for some $p \le 2^{\tOh(n)}$ (with sufficiently many hidden logarithmic factors) and still obtain the same result. This yields a probabilistic circuit $\C'$ over $(\Z_p^{t'},\boxplus,\boxtimes)$. Running \thmref{lokned} on $\C'$ would yield a randomized \SubsetSum algorithm with time $\tOh(n t \log p) = \tOh(n^2 t)$ and space $\tOh(n \log p) = \tOh(n^2)$, if we were able to efficiently compute the required root of unity (which is a non-trivial problem). 

In the next section, we further reduce time and space and we provide an algorithm for computing an appropriate root of unity.

\subsection{Reducing the Domain Size}
\label{sec:reducingdomainsize}

In the last section, we designed a probabilistic circuit $\C$ over $(\N^{t'};\boxplus,\boxtimes)$ and then replaced $\N$ by $\Z_p$ for sufficiently large $p = 2^{\tOh(n)}$.
Now we will pick $p$ as a random prime, chosen uniformly from a set $P$ of size $\tOmega(n^{2})$. This randomness is independent from the randomness of picking a circuit $C \in \C$. Since any number $k \le 2^{\tOh(n)}$ has at most $\tOh(n)$ prime factors,  for any $C \in \C$ the probability that $p$ divides $\outC$ is at most $\tOh(1/n)$. Interpreting a circuit~$C \in \C$ as a circuit over $(\Z_p^{t'};\boxplus,\boxtimes)$, i.e., performing all arithmetic operations of $C$ modulo~$p$, yields a circuit $C_p$ over $(\Z_p^{t'};\boxplus,\boxtimes)$ that computes the value $\outC \bmod p$. Taking into account the randomness of $C \in \C$, we obtain a probabilistic circuit $\C_p$ over $(\Z_p^{t'};\boxplus,\boxtimes)$ with the following properties (see \lemref{FSScircuit}):
\begin{enumerate}[label=(\arabic*)]
  \item for any $0 \le s \le t$, $s \not\in \SSB{Z}{t}$ we have $\Pr_{p \in P}[ \Pr_{C_p \in \C_p}[ \out{C_p}[s] = 0]] = 1$, and
  \item for any $0 \le s \le t$, $s \in \SSB{Z}{t}$ we have $\Pr_{p \in P}[ \Pr_{C_p \in \C_p}[ \out{C_p}[s] = 0]] \le \tOh(1/n)$.
\end{enumerate}
Note that \thmref{lokned} and \lemref{FSScircuit} now yield total time $\tOh(|C| t \log(\max P)) = \tOh(n t \log(\max P))$ and space $\tOh(|C| \log(\max P)) = \tOh(n \log(\max P))$ for solving \SubsetSum. 

It remains to choose an appropriate set of primes~$P$.  From the above discussion we have the following requirements, where we need (1) for running DFT, see \thmref{lokned}, (2)  to avoid $\tOh(n)$ prime factors, and (3) since $\max P$ appears in the running time.

\begin{lem} \label{lem:reductionblar}
  For any power of two $\tau$, and $n \le \tau$, there is a set $P$ of primes such that 
  \begin{enumerate}
  \item[(1)] $\Z_p$ contains a $\tau$-th root of unity $\omega_p$ for any $p \in P$, 
  \item[(2)] $|P| = \tOmega(n^{2})$, and 
  \item[(3)] assuming ERH we have $\max P = \tau^{\Oh(1)}$, while unconditionally we have $\max P \le \Oh(\exp(\tau^\eps))$ for any $\eps > 0$ (with the $\polylog(n)$ in $|P| = \tOmega(n^{2})$ depending on $\eps$).
  \end{enumerate} 
  We can pick a random $p \in P$ and compute~$\omega_p$ in randomized time $(\log(\max P))^{\Oh(1)}$, with error probability $1/\poly(\tau)$.
\end{lem}

Recall from \thmref{lokned} that the time and space bounds incur a factor $\Oh(\log p)$. Assuming ERH we have $\Oh(\log p) = \Oh(\log \tau) = \Oh(\log t)$, which yields the desired time $\tOh(n t)$ and space $\tOh(n \log t)$, proving the first statement of \thmref{polynomial}. The unconditional statement follows similarly.
%
%

\begin{proof}
We first simplify requirement (1). For any prime~$p$, by Fermat's little theorem any $a \in \Z_p$, $a \ne 0$, is a $(p-1)$-th root of unity. The additional property $a^{k}\ne 1$ for all $1 \le k < p-1$ of a \emph{primitive} $(p-1)$-th root of unity is satisfied for any generator $g$ of the multiplicative group $\Z_p^* = \Z_p \setminus \{0\}$. Finally, if $\tau$ divides $p-1$ and $g$ is a primitive $(p-1)$-th root of unity then $g^{(p-1)/\tau}$ is a primitive $\tau$-th root of unity. Note that if $\tau$ divides $p-1$ then $p$ is in the arithmetic progression $\{1 + k \cdot \tau \mid k \in \N\}$, which we write as $1+ \tau \cdot \N$.
Hence, we many replace requirement (1) by
\begin{enumerate}[label=(\arabic*')]
  \item $P \subseteq 1 + \tau \cdot \N$.
\end{enumerate}

The above argument shows existence of a primitive $\tau$-th root of unity, but it does not yield an efficient algorithm for determining one, since there is no efficient algorithm known for finding a generator of $\Z_p^*$ (the known algorithms have to compute a prime factorization of $p-1$). We circumvent this problem by using that $\tau$ is a power of two. In this situation, we can compute a $\tau$-th root of unity as follows.

\begin{lem} \label{lem:findingomega}
  Let $p$ be prime and $\tau$ be a power of two, where $\tau$ divides $p-1$. Pick $a \in \Z_p^*$ uniformly at random and set $\omega := a^{(p-1)/\tau}$. If $\omega^{\tau} \ne 1$ or $\omega^{\tau/2} = 1$ then restart this procedure. This method finds a primitive $\tau$\nobreakdash-th root of unity $\omega \in \Z_p$ in expected time $\Oh(\polylog(p))$.
\end{lem}

Note that by halting after $\Oh(\polylog(p))$ iterations we may obtain an algorithm with \emph{worst-case} running time $\Oh(\polylog(p))$ and error probability $1/\poly(p) \le 1/\poly(t)$.

\begin{proof}
  First observe that if the method finishes then $\omega$ is a primitive $\tau$-th root of unity. Indeed, $\omega$ is a $\tau$-th root of unity since $\omega^\tau = 1$. For showing that $\omega$ is primitive, consider the minimal $\ell \ge 1$ with $\omega^\ell = 1$. From $\omega^\tau = 1$ it follows that $\ell$ divides $\tau$, and thus $\ell$ is a power of two. In particular, if $\ell < \tau$ then $\omega^{\tau/2} = \omega^{\ell \cdot \tau/(2\ell)} = 1^{\tau/(2\ell)} = 1$, contradicting the break condition $\omega^{\tau/2} \ne 1$. Hence, $\ell = \tau$, and for any $1 \le k < \tau$ we have $\omega^k \ne 1$, showing that $\omega$ is primitive.
  
  Since arithmetic operations in $\Z_p$ can be performed in time $\Oh(\polylog(p))$, it remains to bound the success probability of one round of the method by $\Omega(1/\polylog(p))$. To this end, it suffices to show that a random $a \in \Z_p^*$ is likely to be a generator of $\Z_p^*$, since then $a$ is a primitive $(p-1)$-th root of unity and $\omega = a^{(p-1)/\tau}$ is a primtive $\tau$-th root of unity. For bounding the probability of $a$ being a generator, we combine the well-known fact that the number of generators is $\phi(p-1)$, where $\phi$ is Euler's totient function, and the asymptotic lower bound $\phi(n) = \Omega(n / \log \log n)$. 
  \IfSODA{\qed}
\end{proof}
It remains to find a set $P \subseteq 1 + \tau \cdot \N$ of $\tOmega(n^2)$ primes. Dirichlet's theorem shows that the arithmetic progression $1 + \tau \cdot \N$ contains infinitely many primes. However, in order to bound $\max P$ we need a quantitative version of Dirichlet's theorem.
In particular, let $\pi(x,\tau)$ be the number of primes in $(1 + \tau \cdot \N) \cap [0,x]$. We need an upper bound on the minimal $x$ such that $\pi(x,\tau) = \tOmega(n^2)$.

The best unconditional bound is given by (a well-known corollary of) the Siegel-Walfisz theorem~\cite{Walfisz1936}: For any constant $C > 0$ there is a constant $C' > 0$ such that if $\tau \le (\log x)^C$ then
$$ \pi(x,\tau) = \textup{Li}(x)/ \phi(\tau)  \pm \Oh(x \exp(-C' \sqrt{\log x})), $$
where $\textup{Li}$ denotes the offset logarithmic integral. Using $\textup{Li}(x) = \Theta(x / \log x)$ and $\phi(\tau) \le \tau \le (\log x)^C$, we obtain $\textup{Li}(x)/ \phi(\tau) = \tOmega(x)$, which dominates the error term $\Oh(x \exp(-C' \sqrt{\log x}))$ for any sufficiently large $x$. Hence, we obtain $\pi(x,\tau) = \tOmega(x)$ for $\tau \le (\log x)^C$. Setting $C = 1/\eps$ and rearranging, we have $\pi(x,\tau) = \tOmega(x)$ for $x \ge \exp(\tau^\eps)$.

Hence, for the threshold $x := n^2 + \exp(\tau^\eps) = \Oh(\exp(\tau^\eps))$ we have $\pi(x,\tau) = \tOmega(n^2)$. We now set $P$ as the set of primes in $R := (1 + \tau\cdot \N) \cap [0,x]$. Clearly, we have $|P| = \tOmega(n^2)$. 
Moreover, note that the set $R$ contains $(x-1)/\tau \le x$ numbers, $\tOmega(x)$ of which are prime. Thus, a random number in $R$ is prime with probability $\Omega(1/\polylog(x)) = \tau^{-\Oh(\eps)}$. Now to pick a random $p \in P$, we repeatedly pick a random number $r \in R$ until it is prime. Recall that checking primality can be done in polynomial time $\Oh(\polylog(x)) = \tau^{\Oh(\eps)}$~\cite{agrawal2004primes}. The expected number of repetitions until we find a prime is also $\tau^{\Oh(\eps)}$. By halting after essentially the same number of repetitions we can also obtain a \emph{worst-case} time bound with sufficiently small error probability, say~$1/\tau$.
The total time for finding a random prime $p \in P$ is thus $\tau^{\Oh(\eps)}$. This procedure has to be performed only once, so this time is negligible compared to the remaining running time $\tOh(nt)$.  
Note that the number of log-factors in the bound $|P| = \tOmega(n^2)$ depends on $\eps$, and since the constants in the Siegel-Walfisz theorem are ineffective, we cannot give an explicit bound on the number of log-factors in terms of~$\eps$.

Assuming the Extended Riemann Hypothesis, very good asymptotic bounds on $\pi(x,\tau)$ are known. 
Using a bound by Titchmarsh, in~\cite[Proposition 4.3]{de2013fast} it is shown, along the lines of the last two paragraphs, that already for threshold $x := \tOh(\tau(\tau + n^2))$ (with sufficiently many hidden logarithmic factors) the set $R$ contains $\tOmega(n^2)$ primes, and sampling random elements of $R$ yields a prime in expected time $\Oh(\polylog(x)) = \Oh(\polylog(\tau))$. This finishes the proof.
\IfSODA{\qed}
\end{proof}

\section{Unbounded Subset Sum}
\label{sec:unbounded}

We briefly discuss the unbounded variant of \SubsetSum, where each input number can be chosen arbitrarily often (whereas in the bounded variant that we studied so far each input number can be chosen at most once). Specifically, given a set $Z$ of $n$ positive integers and target $t$, the task is to determine whether any sequence over $Z$ sums to exactly $t$. Note that here it makes no sense to define $Z$ as a multi-set and thus $n \le t$. We present a simple $\tOh(t)$ algorithm.

First consider the problem whether any sequence over $Z$ of length at most $k$ sums to $t$. Clearly, this problem can be solved in time $\tOh(k t)$, by computing the $k$-fold sumset $Z \plt Z \plt \ldots \plt Z$, and checking whether it contains $t$. Note that for the usual bounded variant of \SubsetSum this property breaks down,  since already the sumset $Z \plt Z$ is not necessarily contained in $\SSB{Z}{t}$, as it contains sums of the form $z+z$ for $z \in Z$, so we have to resort to sumsets $Z_1 \plt \ldots \plt Z_k$ over partitionings $Z = Z_1 \cup \ldots \cup Z_k$. This is a reason why (bounded) \SubsetSum is much harder than the unbounded version \UnbSubsetSum. 

\algref{unbsubsetsum} solves the \UnbSubsetSum problem for each $t' \le t$ at once, i.e., it computes $\SSunb{Z}{t} := \{0,\ldots,t\} \cap \{ a_1+\ldots+a_k \mid k \in \N_{\ge 0}, \, a_1,\ldots, a_k \in Z  \}$.
We use the fact that the classic dynamic programming algorithm computes $\SSunb{Z}{t}$ in time $O(n t)$. 

\algnewcommand{\LineComment}[1]{\Statex \(\triangleright\) \emph{#1}}
\begin{algorithm}
\caption{\UnbSubsetSum in time $\tOh(t)$. Given $n$ positive integers $Z$ and target~$t$, the algorithm computes for each $t' \le t$ whether a sequence over $Z$ sums to $t'$.}\label{alg:unbsubsetsum}
\begin{algorithmic}[1]
\State compute $S_0 := \SSunb{Z}{t/n}$ using the classic dynamic program in time $\Oh(n \cdot t/n) = \Oh(t)$
\For {$i=1,\ldots,\lceil \log n \rceil$}
  \State $t_i := 2^i t / n$
  \State $S_i := S_{i-1} \pl[t_i] S_{i-1} \pl[t_i] Z$
\EndFor
\State \Return $S_{\lceil \log n \rceil} \cap \{0,\ldots,t\}$
\end{algorithmic}
\end{algorithm}

Since the $i$-th iteration takes time $\tOh(t_i) = \tOh(2^i t / n)$, the total running time of \algref{unbsubsetsum} is $\tOh(t)$, more precisely $\Oh(t \log t)$. For correctness, we argue inductively that $S_i = \SSunb{Z}{2^i t/n}$, and thus $S_{\lceil \log n \rceil} \cap \{0,\ldots,t\} =\SSunb{Z}{t}$. 
Consider any sequence $A$ over $Z$ summing to an integer in $[0,2^i t/n]$. If $\SUM(A) \le 2^{i-1} t/n$, then 
\begin{align*}
  \SUM(A) \in & \SSunb{Z}{2^{i-1} t/n} \IfSODA{\\ &} = S_{i-1} \subseteq S_{i-1} \pl[t_i] S_{i-1} \pl[t_i] Z = S_i.
\end{align*}
Thus, assume $\SUM(A) > 2^{i-1} t/n$.
We split $A = (a_1,\ldots,a_k)$ at the smallest index $j$ with $a_1+\ldots+a_j > 2^{i-1} t/n$ into $A_1 = (a_1,\ldots,a_{j-1})$, $a_j$, and $A_2 = (a_{j+1},\ldots,a_k)$. Observing that $a_j \in Z$ and $\SUM(A_1),\SUM(A_2) \le 2^{i-1} t/n$ and hence $\SUM(A_1),\SUM(A_2) \in S_{i-1}$, we obtain $\SUM(A) \in S_{i-1} \pl[t_i] S_{i-1} \pl[t_i] Z = S_i$. Hence, $\SSunb{Z}{2^i t/n} \subseteq S_i$. Correctness follows, since clearly the converse $S_i \subseteq \SSunb{Z}{2^i t/n}$ holds as well.

\section{Conclusion}

A textbook algorithm solves \SubsetSum in pseudopolynomial time $\Oh(n t)$ on inputs consisting of $n$ positive numbers with target $t$. As our main result we present an improved algorithm running in time $\tOh(n + t)$. 
This improves upon a classic algorithm and is likely to be near-optimal, since it matches conditional lower bounds of $t^{1-\eps}$ for any $\eps > 0$ from \SetCover and combinatorial \kClique. 
Our algorithm heavily uses randomization and has one-sided error. The main tricks are to find \emph{small} solutions using color-coding, and to split \emph{large} solutions by a two-stage color-coding-like process.

We also improve the best known polynomial space algorithm for \SubsetSum to time $\tOh(nt)$ and space $\tOh(n)$ assuming the Extended Riemann Hypothesis, and to time $\tOh(n \, t^{1+\eps})$ and space $\tOh(n \, t^\eps)$ for any $\eps > 0$ unconditionally. This improvement is achieved by combining our $\tOh(n+t)$ algorithm with the previously best polynomial space algorithm by Lokshtanov and Nederlof~\cite{lokshtanov2010saving}, and by working modulo a random prime in an arithmetic progression, which explains the connection to ERH. We leave it as an open problem to obtain an algorithm with time $\tOh(n+t)$ and space $\tOh(n)$.

Our techniques are very different from the recent $\tOh(\sqrt{n} \, t)$ algorithm by Koiliaris and Xu~\cite{KoiliarisX15}, which is the fastest known deterministic algorithm.   


\paragraph{Acknowledgements}
The author wants to thank Marvin K\"unnemann and Jesper Nederlof for providing useful comments on a draft of this paper.

\end{document}